\pdfoutput=1

\documentclass[conference]{IEEEtran}
\IEEEoverridecommandlockouts
\usepackage{cite}
\usepackage{amsmath,amssymb,amsfonts}
\usepackage{algorithmic}
\usepackage{graphicx}
\usepackage{textcomp}
\usepackage{xcolor}
\def\BibTeX{{\rm B\kern-.05em{\sc i\kern-.025em b}\kern-.08em
    T\kern-.1667em\lower.7ex\hbox{E}\kern-.125emX}}

\usepackage[
font=small,labelfont=bf
]{caption}
 \usepackage{flushend}
 \flushend
\usepackage{amsfonts, amssymb, amsthm, color,mathtools}
\usepackage{bbm}
\usepackage{comment}
\usepackage[inline]{enumitem}
\usepackage{thmtools}
\usepackage{thm-restate}

\usepackage{tikz}
\usepackage{svg}
\usepackage{amsmath}

\usepackage[capitalize]{cleveref}
\usepackage{caption}
\usepackage{subcaption}
\usepackage{dsfont}
\usepackage[nolist]{acronym}
\usepackage{mathtools} %
\usepackage[scientific-notation=true]{siunitx}
\usepackage{float}
\usepackage{mathrsfs}

\usepackage{tcolorbox,tikz, lipsum}
\usetikzlibrary{backgrounds,calc,shadings,shapes.arrows,shapes.symbols,shadows,fit,positioning,topaths,hobby,shapes.geometric,shapes,snakes,patterns,arrows,automata,plotmarks}
\tcbuselibrary{skins}
\tcbuselibrary{breakable}
\usepackage{svg}
\usepackage{pgfplots}
\pgfplotsset{compat=1.10} 
\usepgfplotslibrary{fillbetween}
\usetikzlibrary{patterns}
\usetikzlibrary{external}

\newtheorem{theorem}{Theorem}

\newtheorem{lemma}[theorem]{Lemma}

\newtheorem{corollary}{Corollary}
\newtheorem{construction}{Construction}

\newtheorem{observation}{Observation}
\theoremstyle{definition}
\definecolor{britishracinggreen}{rgb}{0.0, 0.26, 0.15}

\crefname{paragraph}{paragraph}{paragraphs}
\Crefname{paragraph}{Paragraph}{Paragraphs}

\newcommand{\define}{\triangleq}

\begin{document}

\title{Scalable and Reliable Over-the-Air\\Federated Edge Learning \\
\author{%
  \IEEEauthorblockN{Maximilian Egger, Christoph Hofmeister, Cem Kaya, Rawad Bitar and Antonia Wachter-Zeh}
  \IEEEauthorblockA{%
                   School of Computation, Information and Technology, Technical University of Munich, Munich, Germany \\\{maximilian.egger, christoph.hofmeister, cem.kaya, rawad.bitar, antonia.wachter-zeh\}@tum.de \vspace{-.7cm}}
    \thanks{This project has received funding from the German Research Foundation (DFG) under Grant Agreement Nos. BI 2492/1-1 and WA 3907/7-1. The paper is submitted to IEEE GLOBECOM 2024.}
}}

\maketitle

\begin{abstract}
Federated edge learning (FEEL) has emerged as a core paradigm for large-scale optimization. However, FEEL still suffers from a communication bottleneck due to the transmission of high-dimensional model updates from the clients to the federator. Over-the-air computation (AirComp) leverages the additive property of multiple-access channels by aggregating the clients' updates over the channel to save communication resources. While analog uncoded transmission can benefit from the increased signal-to-noise ratio (SNR) due to the simultaneous transmission of many clients, potential errors may severely harm the learning process for small SNRs. To alleviate this problem, channel coding approaches were recently proposed for AirComp in FEEL. However, their error-correction capability degrades with an increasing number of clients. We propose a digital lattice-based code construction with constant error-correction capabilities in the number of clients, and compare to nested-lattice codes, well-known for their optimal rate and power efficiency in the point-to-point AWGN channel.
\end{abstract}


\section{Introduction}

In federated edge learning (FEEL), clients use their individual data to collaboratively train a machine learning (ML) model centrally orchestrated by a federator, often through the use of iterative optimization processes such as (stochastic) gradient descent \cite{mcmahan2017communication}. While the federator's model can be broadcast to the clients at each iteration, the local updates computed by the clients are usually transmitted over a noisy channel through separation-based multiple access schemes such as Orthogonal Frequency-Division Multiple Access (OFDMA). Hence, each device requires a separate communication resource, making the communication complexity scale linearly with the number of clients. Since the communication complexity is crucial for high-dimensional ML models, many proposed communication-efficient FEEL methods leverage the source-channel separation principle by compressing the model updates or the gradients and assuming a reliable and error-free end-to-end communication, e.g., \cite{deepGC,QSGD,signSGD}.

However, individual updates are not required for FEEL. It suffices for the federator to observe the aggregated local updates. This facilitates over-the-air computation (AirComp) in FEEL, which makes use of signal superposition in wireless multiple-access channels (MACs) \cite{ml_wireless_edge,yang2020federated,zhu2021over}. It is known from \cite{harnessing_interference} that AirComp allows the computation of nomographic functions; %
and aggregations in FEEL are a special case of those. The clients simultaneously\footnote{We assume that the clients can synchronize their transmission to leverage the additive properties of the channel. This synchronization faces its own practical challenges that are out of the scope of this work.} transmit their local updates and the federator decodes the sum, allowing for better transmission rates than separation-based approaches. See Figure~\ref{fig:simple_oac} for a simple illustration of the AirComp setting.
Analog uncoded transmission without error correction is a simple and provably optimal solution for high signal-to-noise ratios (SNRs)  when the number of available channel uses $n$ equals the gradient dimension $k$ \cite{gastpar_uncoded_optimal}. However, it can suffer from significant performance losses, incurring non-negligible errors in the trained machine learning model \cite{yang2022over}. %
In a limited bandwidth regime where $n<k$, each element of the gradient can be quantized to a single bit, and quadrature amplitude modulation (QAM) can be used for digital modulation with majority-vote-based decoding at the federator \cite{obda}. An improved scheme using finer quantization has been proposed in \cite{qiao2023unsourced}. In analog systems, edge devices can sparsify their gradients and project them to a lower dimensional vector before transmission \cite{ml_wireless_edge}. %

Since transmission errors, especially for low SNRs, can significantly harm the learning process, we focus on cases where $n>k$, and hence, error correction is possible. The reliable reconstruction of a function of sources over a MAC was first analyzed in \cite{comp_over_mac} using lattice codes from \cite{analog_match}, proving the suboptimality of source-channel separation. Separation was shown to be exponentially suboptimal as the number of Gaussian sources increases \cite{gastpar_uncoded_optimal}. %
In \cite{lattices_good}, the existence of a lattice sequence both good for source and channel coding is proven. Nested lattices are used in  \cite{nomographic_efficient,compute_and_forward,local_gossip} due to their linear structure and good algebraic properties. %
They can asymptotically achieve the capacity of the point-to-point Additive White Gaussian Noise (AWGN) channel \cite{achieving_capacity}, and the reliable decoding of the sum of lattice codewords allows retrieving the \textit{modulo-sum} of the messages simultaneously sent by $K$ devices. However, determining the sum without wraparound is paramount in FEEL. %

To allow digital modulation, %
the authors of \cite{channelcomp} formulate an optimization problem to find constellations without destructive overlaps and facilitate arbitrary finite function computations. %
In~\cite{balanced_ofdm},  balanced number systems are used in conjunction with the on-off keying of OFDM subcarriers.
In contrast, we will use balanced numbers with lattice coding for error correction purposes. 
Recently, an AirComp FEEL scheme based on OFDM combined with convolutional and LDPC codes was proposed %
\cite{you2023broadband}, whose decoding complexity was reduced using non-binary LDPC codes and nested lattices \cite{xie2023joint}. %
However, to ensure the sum of the received symbols maps to the correct codeword, %
the finite field size of the channel code must scale linearly with $K$, leading to a significant increase in the block error rate with $K$. 
Such a drawback is inherent to most code constructions over finite fields for AirComp in FEEL. 
To mitigate this, we propose two lattice-based code constructions and analyze them regarding their scaling properties in $K$.

\begin{itemize}
    \item We propose in \cref{sec:balanced_numbers} a lattice-based construction for digital transmission %
    achieving bandwidth expansion via a balanced numeral representation. %
    We derive an upper bound on the expected error that is independent of $K$.
    \item For analog transmission, we investigate in \cref{sec:nested_lattices} known nested lattices %
    for finite field addition over the AWGN channel. We give an upper bound on the average error rate and show that the achievable asymptotic rate and the error probability do not scale well in $K$.
    \item We compare the error-correction capabilities and highlight the regimes of benefit. While nested lattices are preferable for small $K$ or very large $n$, lattices with balanced numbers are superior in the regime of interest. %
\end{itemize}

\section{Over-the-Air Federated Learning Model}

For two integers $\tau, \zeta$, with $\tau < \zeta$, let $[\tau]$ denote the set of integers $\{1, \dots \tau\}$, and $[\tau, \zeta]$ denote $\{\tau, \dots \zeta\}$. 
Let $K$ be the number of clients participating in the training process, enumerated by $i \in [K]$. At each iteration $t$, the clients obtain the global model $\theta_t$ from the federator via a reliable channel and compute a $k$-dimensional gradient $\mathbf{g}_i(\theta_t) \in \mathbb{R}^k$ of a pre-defined loss function evaluated at their data $\mathcal{D}_i$. With a learning rate $\eta>0$, the federator updates the global model by averaging the clients' gradients, i.e., $\theta_{t+1} = \theta_t - \eta \sum_{i=1}^K \mathbf{g}_i(\theta_t)$, and the algorithm repeats until certain performance indicators are reached. Let $\mathbb{Z}_q$ denote the ring of intergers $\{0, \dots q-1\}$ modulo $q$. Each client employs a quantizer $Q: \mathbb{R}^k \rightarrow \mathbb{Z}_q^k$ to obtain a discrete representation of the gradient, i.e., $\tilde{\mathbf{g}}_i(\theta_t) = Q(\mathbf{g}_i(\theta_t))$. Using a channel encoder $\phi: \mathbb{Z}_q^k \rightarrow \mathbb{R}^n$, those gradients are mapped to the transmit signal $\mathbf{X}_i(t) = \phi(\tilde{\mathbf{g}}_i(\theta_t))$. We assume the clients have channel state information (CSI) and are perfectly synchronized, and transmit their signals $\mathbf{X}_i(t)$ over an AWGN channel, such that the federator receives $\mathbf{Y}(t) = \sum_{i=1}^K \mathbf{X}_i(t) + \mathbf{N}(t)$, where $\mathbf{N}(t) \sim \mathcal{N}(\boldsymbol{0}, P_N \mathbf{I}_{n \times n})$ is independent and identically distributed (i.i.d.) Gaussian noise with power $\lim_{n\rightarrow \infty} \frac{1}{n} \Vert \mathbf{N}(t) \Vert_2^2 = P_N$. The goal of the federator is to decode the sum of the gradients $\hat{\mathbf{g}}(\theta_t) \define \sum_{i=1}^K \tilde{\mathbf{g}}_i(\theta_t)$ by applying the channel decoder $\phi^{-1}: \mathbb{R}^n \rightarrow \mathbb{R}^k$, and update the model as $\theta_{t+1} = \theta_{t} - \frac{\eta}{K} \hat{\mathbf{g}}(\theta_t)$, with learning rate $\eta$. 
\begin{figure}[!t]
    \centering
    \tikzset{every picture/.style={line width=0.75pt}} %
    \resizebox{\linewidth}{!}{ \input{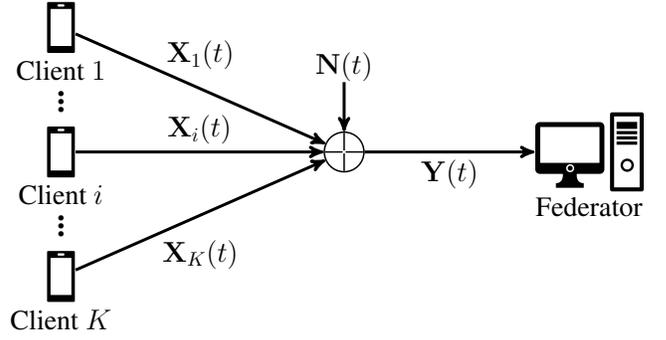} }
    \caption[A simplified illustration of an over-the-air federated learning scheme.]{A simplified illustration of an over-the-air federated learning setting.
    The $K$ edge devices transmit simultaneously to the federator.
    The federator observes the noisy sum.}
    \label{fig:simple_oac}
\end{figure}
To ensure model convergence, we require the gradient estimate to be unbiased, i.e, $\mathbb{E}[\hat{\mathbf{g}}(\theta_t)] = \mathbb{E}\left[\sum_{i=1}^K \mathbf{g}_i(\theta_t)\right]$.
By linearity of expectation, this is achieved by ensuring that $\mathbb{E}\left[\tilde{\mathbf{g}}_i(\theta_t)\right] = \mathbb{E}\left[ \mathbf{g}_i(\theta_t)\right]$ for all $i\in[K]$, and all $t$. In the following, we will be concerned with a single iteration and drop the iteration index $t$ from our notation. 
We use a variant of stochastic rounding~\cite{gupta2015deep,croci2022stochastic}, known to be unbiased. Let $g_{\textrm{min}}$ and $g_{\textrm{max}}$ be the minimum and the maximum value that a gradient can take, respectively. Each entry $g_{i,j}$ of the gradient $\mathbf{g}_i$ is quantized separately according to %
\begin{equation*}\label{eq:qsgd_definition}
    Q_{\textrm{stoc}}(g_{i, j}, q) = \xi\left((q-1) \frac{g_{i, j} - g_\textrm{min}}{g_\textrm{max} - g_\textrm{min}}\right), 
\end{equation*}
with \vspace{-.1cm}
\begin{equation*}\label{eq:}
    \xi(x) =
    \begin{cases}
         \lceil  x \rceil    &  \text{w.p. } \frac{x - \lfloor x \rfloor}{\lceil x \rceil - \lfloor x \rfloor} \\
         \lfloor  x \rfloor    &  \text { otherwise}
    \end{cases}
    .\end{equation*}
Intuitively, we pick $q$ linearly spaced quantization levels and probabilistically round to the quantization level above or below $g_{i,j}$ with probability inversely proportional to the distance to the quantization level. 
The corresponding dequantizer also acts entry-wise and is given by \vspace{-.1cm} \begin{align*}
    Q^{-1}_\textrm{stoc}(x, q) = K g_\textrm{min} + (g_\textrm{max} - g_\textrm{min}) \frac{x}{q-1}.
\end{align*}

It is straightforward to prove the following lemma.
\begin{lemma}
    The stochastic quantizer $Q_{\textrm{stoc}}$ preserves the expectation of the gradient, i.e. 
    \begin{align*}
        \mathbb{E}\Big[Q_\textrm{stoc}^{-1}\Big(\sum_{i \in [K]} Q_{\textrm{stoc}}\big(\mathbf{g}_i(\theta_t)\big)\Big)\Big] = \mathbf{g}(\theta_t),
    \end{align*} where the expectation is over the randomness of the quantizer. 
\end{lemma}

Let $\tilde{g}_{i,j}$ subsequently be the elements of the quantized gradient $\tilde{\mathbf{g}}_i(\theta)$. We will focus on different constructions of the encoder $\phi$, where we rely on the following lattice preliminaries. 
An $n$-dimensional lattice $\Lambda, n \in \mathbb{Z^+}$, is a discrete additive
subgroup of $\mathbb{R}^n$, such that $\forall \mathbf{u},\mathbf{v} \in \Lambda: \mathbf{u} + \mathbf{v}
\in \Lambda$. It is defined by a generator matrix $\mathbf{B} \in \mathbb{R}^{n \times n}$ as
$
    \Lambda = \{ \mathbf{B} \mathbf{v} \mid \mathbf{v} \in \mathbb{Z}^n \}.
$
A \textit{Euclidean nearest neighbor quantizer} $Q_{\Lambda}(.)$ for the lattice $\Lambda$ maps every vector $\mathbf{v} \in \mathbb{R}^n$ to the nearest lattice point in terms of the Euclidean distance, i.e.,
$
    Q_{\Lambda}(\mathbf{v}) = \arg \min_{\boldsymbol{\lambda} \in \Lambda} \Vert\mathbf{v}-\boldsymbol{\lambda}\Vert_2 .
$
The fundamental Voronoi region $\mathcal{V}$ of $\Lambda$ is the set of points in the
$n$-dimensional space whose closest lattice point is the zero word, i.e.,
$
    \mathcal{V} = \{ \boldsymbol{v} \mid Q_{\Lambda}(\mathbf{v}) = \mathbf{0} \}.
$
The Voronoi cell associated with a lattice point $\boldsymbol{\lambda} \in \Lambda$ can be
calculated by shifting the fundamental Voronoi cell $\mathcal{V}$ by
$\mathbf{\lambda}$. The modulo operation with respect to a lattice $\Lambda$ shifts any vector $\mathbf{v} \in \mathbb{R}^n$ into $\mathcal{V}$, i.e.,
$
    \mathbf{v} \textrm{ mod } \Lambda = \mathbf{v} -  Q_{\Lambda} (\mathbf{v}).   
$

\section{Balanced Numeral Lattice Codes} \label{sec:balanced_numbers}

While nested lattices are a viable channel coding solution for AirComp due to their linearity and optimal power usage in point-to-point transmission, i.e., when $K=1$, they suffer from a major weakness: the suboptimal scaling properties in the number $K$ of devices, making this approach infeasible in settings such as large-scale FEEL. We defer the reader to \cref{sec:nested_lattices} for the details.

We propose a simple yet effective lattice construction based on a balanced numeral representation suitable for settings with many devices in the regime of finite block-lengths $n$.  The idea is to represent the quantized gradients $\mathbf{g}_i$ using a balanced number system, which is a natural choice given that gradients can take both negative and positive real values. 

\begin{construction}[Balanced numeral representation]
    Let $\mathbb{S}_{\beta} = \left[-\frac{\beta-1}{2}, \frac{\beta-1}{2}\right]$ be the set of digits for a balanced number system with odd integer base $\beta>1$, and $\alpha>0$ a positive integer. Any integer $z \in \left[-\frac{\beta^\alpha-1}{2},\frac{\beta^\alpha-1}{2} \right]$ can be    represented by $\alpha$ digits $\mathbf{d} = (d_{\alpha-1},\dots,d_0)_{\bar{\beta}}$, $d_\ell \in \mathbb{S}_\beta$ such that\footnote{The bar notation is to avoid confusion between the base-$\beta$ representation and the balanced numeral representation.}
\begin{equation*}\label{eq:}
    z = (d_{\alpha-1},\dots,d_0)_{\bar{\beta}} = \sum_{\ell=0}^{\alpha-1}{d_\ell \beta^\ell}
.\end{equation*}
\end{construction}
Since the summation will result in digits in $\mathbb{Z}$, we allow a relaxed definition of the resulting sum, whose representation is non-unique. Let now $k \vert n$, and let $\alpha \define n/k \in \mathbb{Z}^+$, then for each entry of the gradient $g_{i,j}$ we use the quantization scheme $Q_{\textrm{stoc}}(g_{i,j}, q)$ proposed above with $q \leq \beta^\alpha$, and a symmetric quantization interval $[g_{\textrm{min}}, g_{\textrm{max}}] = [-g_b, g_b]$.

Since all operations are element-wise, let $\mathbf{d}_{i,j} \define (d_{i,j,\alpha-1},\dots,d_{i,j,0})_{\bar{\beta}}$, $i \in [K], j \in [k]$ be the $\alpha$-digit balanced numeral representation for device $i$ and gradient dimension $j$. We define the following encoder $\mathcal{E}(\cdot)$ to construct the transmit signal $\mathbf{X}_i \in \mathbb{R}^{\alpha}$ of
device $i$ as
\begin{equation*}\label{eq:balanced_xi_definition}
    \mathbf{X}_i = \mathcal{E}(\mathbf{d}_{i,j}) = \frac{\sqrt{P_X}}{(\beta-1)/2}\sum_{\ell=1}^{\alpha} \mathbf{e}_{\ell} d_{i,j, \ell-1}
,\end{equation*}
where $\mathbf{e}_{\ell}$ are standard unit vectors in $\alpha$ dimensions. The code is thus defined as $\{\mathcal{E}(\mathbf{d}) \vert \mathbf{d} \in \mathbb{S}_\beta^\alpha\}$. The factor $\frac{\sqrt{P_X}}{(\beta-1)/2}$
scales the transmit signal per channel use such that it never exceeds $P_X$. While imposing assumptions on the distribution of the gradients would enable a more thorough use of the power budget, we conduct our analysis for the general case where gradient distributions are unknown a priori. The transmit signals $\mathbf{X}_i, i \in [K]$, are all part of a scaled $\alpha$-dimensional integer lattice $\mathbb{Z}^{\alpha}$ with generator matrix $\mathbf{B} \define \frac{\sqrt{P_X}}{(\beta-1)/2} (\mathbf{e}_1, \cdots, \mathbf{e}_\alpha)^T$, termed $\Lambda$, with a Voronoi region in the shape of an $\alpha$-dimensional cube.%

The federator observes the noisy sum $\mathbf{Y} = \sum_{i=1}^{K}{\mathbf{X}_i} + \mathbf{N}$. Since all $\mathbf{X}_i$ are members of this lattice, so is their sum $\sum_{i=1}^{K}{\mathbf{X}_i}$. The decoding is equivalent to that of the unconstrained AWGN channel. Hence, the ML decoder matches the lattice decoder; in this case, the Euclidean nearest neighbor quantizer $Q_{\Lambda}(\mathbf{Y})$. With $\hat{\mathbf{d}}_{j}=(\hat{d}_{j,0},\dots,\hat{d}_{j,\alpha-1}) \define \frac{(\beta-1)/2}{\sqrt{P_X}} Q_{\Lambda}(\mathbf{Y})$, the receiver applies the decoder
\begin{equation*}\label{eq:balanced_digit_to_int}
    \hat{g}_j = \mathcal{E}^{-1}(\mathbf{Y}) = \sum_{\ell=0}^{\alpha-1}{\beta^\ell \hat{d}_{j,\ell}},
\end{equation*}
and the inverse quantizer $Q^{-1}_\textrm{stoc}(\hat{g}_j, q)$. 
Note that both the encoder $\mathcal{E}$ and the decoder $\mathcal{E}^{-1}$ are linear.

An error occurs when $g_j \define \sum_{i=1}^K g_{i,j} \neq \sum_{\ell=0}^{\alpha-1}{\beta^\ell \hat{d}_{j,\ell}}$. If the Gaussian
noise vector $\mathbf{N}$ does not leave the fine cubic cell $\mathcal{V}$ of $\Lambda$, the lattice decoder maps back to the correct lattice point. However, if $\mathbf{N}$ leaves this region, the decoder might or might not decode back to the correct sum. If the added noise maps to a zero-valued codeword, the correct result is decoded despite $\mathbf{N}$ leaving $\mathcal{V}$. This is reflected in the following theorem.

\begin{theorem}
    Consider a single dimension of the gradient. Let 
    $N \sim \mathcal{N}(0,P_N)$ and
    \begin{equation*}\label{eq:}    
        \mathbb{D} = 
        \left\{
        \mathbf{\hat{d}}=(\hat{d}_{0},\dots,\hat{d}_{\alpha-1})
        \bigg|
        \sum_{\ell=0}^{\alpha-1}{\beta^\ell \hat{d}_{\ell}} = 0 
        \right\}
    .\end{equation*}
    The dimension-wise error probability with Balanced Numeral codes is 
    \begin{equation*}\label{eq:balanced_prob_err}
        P_e^j 
        = 
        1-\sum_{\mathbf{\hat{d}} \in \mathbb{D}}
        \!\! \left( \prod_{\ell=1}^{\alpha}{\! \Pr\!\left\{ \frac{N (\beta-1)}{\sqrt{P_X}} \! \in \!\! 
            \left[
                2 \hat{d}_{\ell-1}\!-\!1
                ,
                2 \hat{d}_{\ell-1}\!+\!1
            \right]
         \right\}} \! \right).
    \!,\end{equation*}
\end{theorem}
\begin{proof}
    Let $\mathbf{N}$ be the i.i.d. additive Gaussian noise with $\mathbf{N} \sim
    \mathcal{N}(0,P_N \mathbf{I}_{\alpha \times \alpha})$. Consider a member
    $\mathbf{X}$ of the $\alpha$-dimensional cubic lattice, which is decoded as
    the integer $\hat{d}$. %
    For any member $\hat{\mathbf{d}}$ of $\mathbb{D}$, the lattice point $\mathbf{X} + \mathcal{E}(\hat{\mathbf{d}})$ also decodes to $\hat{d}$ since
    $\mathcal{E}(\hat{\mathbf{d}})$ has the decoding result
    $\sum_{\ell=0}^{\alpha-1}{\beta^\ell \hat{d}_{\ell}} = 0$. This is true by
    linearity of the encoder and decoder. %
    Thus, if the Gaussian noise vector
    $\mathbf{N}$ decodes to $0$, we get the correct result. If not, an error
    occurs. The probability that $\mathbf{N}$ decodes to a given vector
    $\mathbf{X}$ is the probability that it is in the fine cubic cell of
    $\mathbf{X}$:
    \begin{equation*}\label{eq:}
        P_e = 1-\sum_{\mathbf{\hat{d}} \in \mathbb{D}}\Pr\{\mathbf{N} \textrm{ is in the cubic cell of } \mathcal{E}(\mathbf{\hat{d}})\}
    \end{equation*}
    Notice that the cubic cell of $\mathbf{\hat{d}}$ is an $n$-dimensional cube
    centered at
    $\mathcal{E}(\mathbf{\hat{d}})=\frac{\sqrt{P_X}}{(\beta-1)/2}\sum_{\ell=1}^{\alpha}
    \mathbf{e}_{\ell} \hat{d}_{\ell-1}$ with edges of length
    $\frac{\sqrt{P_X}}{(\beta-1)/2}$. For $\mathbf{N}$ to be in this cubic cell,
    $N_i, \forall i \in \{0, \cdots, \alpha-1\}$, must be at most $\frac{\sqrt{P_X}}{\beta-1}$
    away from the center of the cell. Then, we may rewrite the equation for $P_e$ as in the statement, where we use the product since the Gaussian noise for each dimension is independent. Since $N_i \sim \mathcal{N}(0,P_N)$ we replace $N_i$ with $N \sim \mathcal{N}(0,P_N)$, which concludes the proof.
\end{proof}
\begin{corollary}\label{cor:balanced_upper_bound}
    With $\beta \leq \lceil \sqrt[\alpha]{q} \rceil$  and ignoring each member of $\mathbb{D}$ except the zero vector and considering all $k$ dimensions of the gradient, we get
    \begin{equation*}\label{eq:balanced_prob_err_upper}
         P_e < 1- \left( 1 - 2 F_{\mathcal{N}(0,P_N)}\left(-\frac{\sqrt{P_X}}{\lceil \sqrt[\alpha]{q}-1\rceil} \right)\right)^{k \alpha}.
    \end{equation*}
\end{corollary}
Note that that worst-case error probability tends to zero as $\alpha$ tends to infinity, and is constant in the number of devices $K$.
For increasing block lengths, the probability of error only becomes arbitrarily small if the rate $\frac{1}{\alpha}$ goes to zero, i.e., the scheme does not have an achievable rate in the information-theoretic sense. 
The reason is that the transmitted codewords are contained in a cube, which in turn is contained in a sphere corresponding to the power constraint. For large lattice dimensions $\alpha$, the available power is not used effectively since the cube occupies a small fraction of the sphere.

\section{Nested Lattice Codes}  \label{sec:nested_lattices}

As an alternative to the proposed scheme, the linearity of nested lattices can be exploited to compute the sum of quantized gradients.
A modulo operation on lattices distributes the points generated by a generator matrix $\mathbf{B}$ from a parallelotope to the Voronoi region of a coarser lattice, which for high $n$ approximates a sphere. This leads to a very efficient use of the power budget $P_X$.
However, we will show that due to the modulo operation, sums of codewords, like codewords themselves, lie in  the fundamental Voronoi region $\mathcal{V}$, which limits the number of available input codewords.

To be ``good'' for channel coding, the points of a lattice must be sufficiently far apart.
In addition, such lattices must fulfill a series of properties regarding \textit{packing}, \textit{covering}, and \textit{quantization}. It is known that there exists a sequence of $n$-dimensional lattices $\Lambda_n$ with these properties as $n \rightarrow \infty$ \cite{lattices_good}. We briefly introduce the properties needed for our analysis. 
Let $\mathcal{B}_n(r) \define \{\mathbf{v} \in \mathbb{R}^n \mid \Vert\mathbf{v}\Vert_2 \leq r \}$ be the $n$-dimensional closed ball of radius $r$ centered at the origin. For two sets $\mathcal{X}$ and $\mathcal{Y}$, let $\mathcal{X}+\mathcal{Y}$ denote the sumset $\mathcal{X}+\mathcal{Y} \define \{x+y | x\in \mathcal{X}, y \in \mathcal{Y}\}$. The set $\Lambda + \mathcal{B}_n(r)$ is a \textit{packing} in Euclidean space if for all points $\mathbf{v},\mathbf{w} \in \Lambda, \mathbf{v} \neq \mathbf{w}$, we have $\left(\mathbf{v} + \mathcal{B}_n(r)\right) \cap \left(\mathbf{w} + \mathcal{B}_n(r)\right) = \emptyset$, i.e., two spheres of radius r centered at $\mathbf{v}$ and $\mathbf{w}$ do not intersect. The packing radius is then defined as 
    $r_{\Lambda}^{\textrm{pack}} = \sup \{ r \mid \Lambda + \mathcal{B}_n(r) \textrm{ is a packing.}\}$.

The set $\Lambda + \mathcal{B}_n(r)$ is a covering of Euclidean space if $\mathbb{R}^n \subseteq \Lambda + \mathcal{B}_n(r)$, 
which means each point in n-dimensional space is covered by at least one sphere of radius $r$ centered around a lattice point. The covering radius $r_{\Lambda}^{\textrm{cov}}$ is defined as the smallest such $r$, i.e.,
    $r_{\Lambda}^{\textrm{cov}} = \inf \{ r \mid  \Lambda + \mathcal{B}_n(r) \textrm{ is a covering}\}$. %
The effective radius $r_{\Lambda}^{\textrm{effec}}$ is defined by $\textrm{Vol}(\mathcal{B}_n(r_{\Lambda}^{\textrm{effec}})) = \textrm{Vol}(\mathcal{V})$.
The second moment $\sigma^2$ of a lattice $\Lambda$ is defined as $\sigma^2 (\Lambda) = \frac{1}{\textrm{Vol}(\mathcal{V})} \cdot \frac{1}{n} \int_{\mathcal{V}} \Vert\mathbf{v}\Vert_2^2 d\mathbf{v}$. The \textit{normalized} second moment is then defined as  $G(\Lambda) = 
\frac{\sigma^2(\Lambda)}{{\textrm{Vol}(\mathcal{V})}^{2/n}}$, also known as the MSE distortion measure. 
The normalized second moment of an $n$-dimensional ball is denoted by $G_n^{\star}$.

A nested lattice code $(\Lambda_1,\Lambda)$, with $\Lambda \subset \Lambda_1$, is described by a coarse lattice $\Lambda$ for shaping, and a fine lattice $\Lambda_1$ for coding. The fundamental Voronoi regions of $\Lambda_1$ and $\Lambda$ are denoted by $\mathcal{V}_1$ and $\mathcal{V}$ respectively. By construction, $\textrm{Vol}(\mathcal{V}_1)$ divides $\textrm{Vol}(\mathcal{V})$. The points in the set $\mathcal{C} = \{ \Lambda_1 \textrm{ mod } \Lambda \} = \{\Lambda_1 \cap \mathcal{V}\}$ are the codewords of the nested lattice code. The rate in bits per channel use
is then defined as
\begin{equation*}\label{eq:nested_lattice_rate}
    \mathcal{R} = \frac{1}{n}\log_2 \lvert \mathcal{C} \rvert = \frac{1}{n} \log_2  \lvert \Lambda_1/\Lambda \rvert = \frac{1}{n} \log_2 \left\lvert \frac{\textrm{Vol}(\mathcal{V})}{\textrm{Vol}(\mathcal{V}_1)} \right\rvert
.\end{equation*}
On a high level, the coarse lattice $\Lambda$ ``matches'' the bandwidth of the source to the channel and attains the average power constraint. The fine lattice $\Lambda_1$ distributes the codewords in $n$-dimensional Euclidean space such that codewords are sufficiently ``far'' apart.

It is well-known that in point-to-point communication, nested lattices can, with a slight modification, achieve the capacity $\frac{1}{2} \log_2(1+\text{SNR})$ of the AWGN channel \cite{achieving_capacity}. This is done through the modulo-lattice additive noise (MLAN) channel transformation \cite{erez2005capacity,forney2000sphere,achieving_capacity}, which we generalize 
to multiple devices and the corresponding nested lattices construction following the agenda in \cite{compute_and_forward}.
The scheme we present in this section can be seen as a special case of the relay computation scheme of \cite{compute_and_forward}, where channel state information is available and the interest is purely in the sum of messages instead of a polynomial equation. 
Similar to the MLAN transform in \cite{achieving_capacity}, given a codeword $\mathbf{c}_i \in \mathcal{V}$, device $i$ transmits $\mathbf{X}_i = [ \mathbf{c}_i - \mathbf{U}_i ] \textrm{ mod } \Lambda$, where $\mathbf{U}_i$ are dithers chosen independently and uniformly over $\mathcal{V}$ using shared randomness between transmitters and the receiver. The input alphabet is the fundamental Voronoi region of the lattice $\Lambda$, called the \textit{shaping lattice}. %
The receiver observes the noisy sum $\sum_{i=1}^{K}\mathbf{X}_i + \mathbf{N}$. The observation is multiplied by the MMSE estimator $\alpha$, to be specified below, and the dither random variables $\mathbf{U}_i, i\in \{1,2,\dots,K\}$ are added to the scaled observation. The result is reduced modulo $\Lambda$. We use Loelinger's Construction A \cite{loeliger} to obtain a random nested lattice code $(\Lambda_1,\Lambda)$. Construction A is based on a random linear code over a field $\mathbb{F}_p$ and as a result, the code construction is homomorphic with addition over $\mathbb{F}_p$.
In order to emulate the integer addition, which we require to aggregate quantized gradients, we have to use a prime field $\mathbb{F}_p$ large enough to represent all possible sums of the $K$ devices quantized gradients. This is common to all coding-theoretic approaches over finite fields.

\begin{observation}\label{lem:p_min}
    The prime $p$ must be chosen such that $p > K(q-1)$ to use nested lattice codes for over-the-air computation where $K$ is the number of devices, and $q$ is the quantization parameter.
\end{observation}

We state the achievable rate of the lattice code construction in the following theorem, which can be seen as a special case of~\cite[Theorems~4~and ~5]{compute_and_forward}. Let $\log_2^+(x) = \max\{\log_2(x), 0\}$.

\begin{theorem}\label{th:achievable_rate_nested}
    For any $\epsilon>0$ and sufficiently large n, there exist nested lattice codes
    $(\Lambda_1,\Lambda)$ such that the decoder can decode the sum of quantized
    gradients with average probability of error $\epsilon$ as long as
    \begin{equation*}\label{eq:nested_achievable_rate_alpha}
        \mathcal{R} < \frac{1}{2} \log_{2}^{+}\left( \frac{P_X}{\alpha^2 P_N + K (1-\alpha)^2 P_X} \right)
        ,\end{equation*}
    as $n \rightarrow \infty$. Choosing $\alpha = \frac{K P_X}{P_N + K P_X}$, we get
    \begin{equation*}\label{eq:nested_achievable_rate_noalpha}
        \mathcal{R} < \frac{1}{2} \log_{2}^{+}\left( \frac{P_X}{P_N} + \frac{1}{K} \right).
    \end{equation*}
\end{theorem}

For $K=1$, nested lattices achieve the capacity of the point-to-point AWGN channel and are, therefore, optimal.
However, for an increasing number of devices $K$, the achievable rate does not scale with the total transmit power $K P_X$ but decreases slightly.
Intuitively, the lack of increase is due to the fact that the field $\mathbb{F}_p$ needs to be able to represent all possible sums of quantized gradient values.
Accordingly, all codewords as well as all sums of codewords need to be able to fit within the fundamental Voronoi region of the shaping lattice, which corresponds to the power constraint on a single device, rather than the combined power of all devices.
The rate decreases slightly due to the noise introduced by the dithers.

In the following, we bound the error probability for the nested lattice construction from above.

\begin{theorem}\label{thm:upper_bound_nested}
    The average probability of error for the nested lattice code construction is upper bounded by
    \begin{equation*}
        P_e < \min\left(e^{K \epsilon_1(\Lambda)n} \left(1 - F_{\chi^2(n)} \left(\frac{{r_\Lambda^{\textrm{pack}}}^2}{p^2 P_{Z_{eq}}}\right)\right),1\right),
    \end{equation*}
    where \vspace{-.2cm}
    \begin{equation*}
        \epsilon_1(\Lambda) = \log\left(\frac{r_\Lambda^{\textrm{cov}}}{r_\Lambda^{\textrm{effec}}}\right) + \frac{1}{2} \log(2 \pi e G_n^\star) + \frac{1}{n},
    \end{equation*}
    and \vspace{-.2cm}
    \begin{equation*}\label{eq:}
        P_{Z_{eq}} = \alpha^2 P_N + K (1-\alpha)^2 \left(\frac{r_\Lambda^{cov}}{r_\Lambda^{\textrm{effec}}}\right)^2 P_X
        .\end{equation*}
\end{theorem}

\begin{proof}[Sketch of Proof]
    We rely on the following result, which is a straightforward generalization of the inflated lattice lemma (cf. \cite[Lemma~2]{achieving_capacity}, \cite[Lemmas~6~and~7]{erez2005capacity}).
    \begin{lemma}%
    \label{lem:inflated_lattice_oac}
    The channel from $\mathbf{c}_1,\dots,\mathbf{c}_K$ to $\mathbf{Y^{\prime}}$ is
    equivalent in distribution to the channel
    $
        \mathbf{Y^{\prime}} = [\sum_{i=1}^{K}\mathbf{c}_i + \mathbf{N^{\prime}}] \textrm{ mod } \Lambda
        ,$
    where
    $
        \mathbf{N^{\prime}} = [\alpha \mathbf{N} + (\alpha-1)\sum_{i=1}^{K}\mathbf{U}] \textrm{ mod } \Lambda
        .$
    \end{lemma}
    The average probability of error is $P_e = \textrm{Pr} \{ \mathbf{N^{\prime} \notin
    \mathcal{V}_1} \}$. Let $\mathbf{N^{\prime\prime}} = \alpha \mathbf{N} +
    (1-\alpha)\sum_{i=1}^{K}\mathbf{U}$. This is smaller than or equal to $\textrm{Pr} \{
    \mathbf{N^{\prime\prime} \not\in \mathcal{V}_1} \}$, since the first implies
    the second, but not vice versa. This is due to the modulo operation mapping
    any point outside $\mathcal{V}$ to a point in $\mathcal{V}$ which might or
    might not be in $\mathcal{V}_1$. The following lemma, adapted from~\cite[Lemma~8]{compute_and_forward}, bounds the density of
    $\mathbf{N}^{\prime\prime}$ from above. %
    \begin{lemma}\label{lem:eq_nested_gaussian}
    The density of %
    $\mathbf{N^{\prime\prime}}$
    can be upper bounded by the density of a Gaussian vector
    $\mathbf{Z} \sim \mathcal{N}(0, P_{Z_{eq}}\cdot\mathbf{I}_n) $ such that
    $f_{\mathbf{N^{\prime\prime}}}(\mathbf{x}) < e^{K \epsilon_1(\Lambda) n}
        f_{\mathbf{Z}}(\mathbf{x})$, with $\epsilon_1(\Lambda)$ and $P_{Z_{eq}}$ as above.
    \end{lemma}
    By the application of \cref{lem:eq_nested_gaussian}, we have
    \begin{equation*}
        P_e = \textrm{Pr} \{\mathbf{N^{\prime} \not\in \mathcal{V}_1} \} \leq \textrm{Pr} \{\mathbf{N^{\prime\prime} \not\in \mathcal{V}_1} \} \leq e^{K\epsilon_1(\Lambda) n} \textrm{Pr} \{\mathbf{\mathbf{Z} \not\in  \mathcal{V}_1} \}
    .\end{equation*}
    Now, we bound $\textrm{Pr} \{\mathbf{\mathbf{Z} \not\in \mathcal{V}_1} \}$ from above by observing that if $\mathbf{Z}$ does not leave the packing radius of the coding lattice $\Lambda_1$, it also does not leave the region $\mathcal{V}_1$. Interpreting the coding lattice as a scaled down and diluted version of the shaping lattice, we have $r_{\Lambda_1}^{\textrm{pack}} \geq r_{\Lambda}^{\textrm{pack}} / p$, since we scale down by a factor of $p$ before diluting. Then, we may upper bound $P_e$ by the probability that $\mathbf{Z}$ leaves the n-dimensional sphere
    of radius $r_{\Lambda}^{\textrm{pack}} / p$, which yields the upper bound in our Theorem.
\end{proof}

The upper bound in \cref{thm:upper_bound_nested} can be tightened by relaxing the achievable rate to sacrifice the $\frac{1}{K}$ term. In our regime of interest, where the number of clients $K$ is large, the rate loss incurred thereby gets negligible as $K \rightarrow \infty$. Instead of choosing alpha to minimize the MSE, choosing $\alpha=1$ removes the self-noise of the system incurred due to the dithers, i.e., $e^{K \epsilon_1(\Lambda)n} \rightarrow 1$. This leads to the following error bound.
\begin{corollary}\label{thm:upper_bound_nested_tight}
    The average probability of error for the nested lattice code construction with $\alpha=1$ is
    upper bounded by
    \begin{equation*}
        P_e < 1 - F_{\chi^2(n)} \left(\frac{{r_\Lambda^{pack}}^2}{p^2 P_{N}}\right).
    \end{equation*}
\end{corollary}

Nested lattice codes are known to have good algebraic and geometric properties and allow arbitrarily small average error probabilities as $n\rightarrow \infty$. However, increasing $K$ requires an increase in $p$, thus decreasing $\textrm{Vol}(\mathcal{V}_1)$ and increasing the effective noise. To compensate, we must increase $n$, hence, leading to a logarithmic increase of channel uses. While better than separation-based approaches where $n$ increases linearly with $K$, this drawback is not present in the construction from \cref{sec:balanced_numbers}.

\section{Comparison}
\definecolor{tableaublue}{HTML}{1F77B4}
\definecolor{tableauorange}{HTML}{FF7F0E}
\definecolor{tableaugreen}{HTML}{2CA02C}
\definecolor{tableaured}{HTML}{D62728}
\definecolor{tableaupurple}{HTML}{9467BD}
\definecolor{tableaubrown}{HTML}{8C564B}
\definecolor{tableaupink}{HTML}{E377C2}
\definecolor{tableaugrey}{HTML}{7F7F7F}
\definecolor{tableaulightgreen}{HTML}{BCBD22}
\definecolor{tableaucyan}{HTML}{17BECF}

\newcommand{\fntwozerozero}{plot_SNR_2.00_masterthesis.csv}
\newcommand{\fntwotwofive}{plot_SNR_2.25_masterthesis.csv}
\newcommand{\fntwothreefive}{plot_SNR_2.35_masterthesis.csv}

\begin{figure*}[!t]
    \centering
    \begin{minipage}{0.49\linewidth}
    \begin{tikzpicture}
        \begin{semilogyaxis}
            [
            title= {SNR=2.00 dB},
            xlabel = Number of Devices $K$,
            ylabel = Probability of Error $P_e$,
            ymax=1,
            xmin=1,
            xmax=5,
            legend pos=south east,
            legend style={font=\small},
            height=0.6\textwidth,
            width=\textwidth,            ]
            
            \addplot[mark=none, tableaugreen, dashed, thick] table [x=Num_Devices, y=Median_Nested, col sep=comma] {data/\fntwozerozero};
            \addplot[mark=none, tableaublue, thick] table [x=Num_Devices, y=Median_Balanced, col sep=comma] {data/\fntwozerozero};
            
            \addplot[name path=Max_Nested, mark=none, tableaugreen, draw=none] table [x=Num_Devices, y=Max_Nested, col sep=comma] {data/\fntwozerozero};
            \addplot[name path=Min_Nested, mark=none, tableaugreen, draw=none] table [x=Num_Devices, y=Min_Nested, col sep=comma] {data/\fntwozerozero};
            \addplot[tableaugreen, fill opacity=0.4] fill between[of=Max_Nested and Min_Nested];
            \addplot[name path=Max_Balanced, mark=none, tableaublue, draw=none] table [x=Num_Devices, y=Max_Balanced, col sep=comma] {data/\fntwozerozero};
            \addplot[name path=Min_Balanced, mark=none, tableaublue, draw=none] table [x=Num_Devices, y=Min_Balanced, col sep=comma] {data/\fntwozerozero};
            \addplot[tableaublue, fill opacity=0.4] fill between[of=Max_Balanced and Min_Balanced];
            \legend{Nested Lattice, Balanced Numerals}
        \end{semilogyaxis}
    \end{tikzpicture}
    \end{minipage}
    \begin{minipage}{0.49\linewidth}
    \begin{tikzpicture}
        \begin{semilogyaxis}
            [
            title= {SNR=2.35 dB},
            xlabel = Number of Devices $K$,
            ylabel = Probability of Error $P_e$,
            ymax=1,
            xmin=1,
            xmax=5,
            legend pos=south east,
            legend style={font=\small},
            height=0.6\textwidth,
            width=\textwidth,
            ]
            
            \addplot[mark=none, tableaugreen, dashed, thick] table [x=Num_Devices, y=Median_Nested, col sep=comma] {data/\fntwothreefive};
            \addplot[mark=none, tableaublue, thick] table [x=Num_Devices, y=Median_Balanced, col sep=comma] {data/\fntwothreefive};
           
           \addplot[name path=Max_Nested, mark=none, tableaugreen, draw=none] table [x=Num_Devices, y=Max_Nested, col sep=comma] {data/\fntwothreefive};
           \addplot[name path=Min_Nested, mark=none, tableaugreen, draw=none] table [x=Num_Devices, y=Min_Nested, col sep=comma] {data/\fntwothreefive};
           \addplot[tableaugreen, fill opacity=0.4] fill between[of=Max_Nested and Min_Nested];

           \addplot[name path=Max_Balanced, mark=none, tableaublue, draw=none] table [x=Num_Devices, y=Max_Balanced, col sep=comma] {data/\fntwothreefive};
           \addplot[name path=Min_Balanced, mark=none, tableaublue, draw=none] table [x=Num_Devices, y=Min_Balanced, col sep=comma] {data/\fntwothreefive};
           \addplot[tableaublue, fill opacity=0.4] fill between[of=Max_Balanced and Min_Balanced];

            \legend{Nested Lattice, Balanced Numerals}
        \end{semilogyaxis}
    \end{tikzpicture} 
    \end{minipage}
    \caption[Comparison of experimental probability of error.]{Comparison of experimental probability of errors of  Nested
    Lattice and Balanced Numeral codes for AirComp.}
    \label{fig:result_experimental}
\end{figure*}
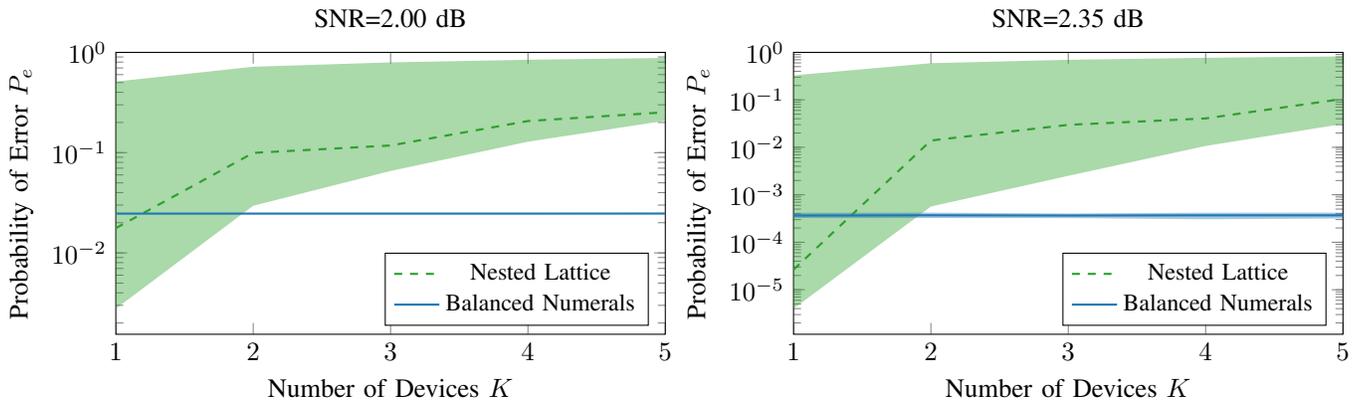

The theoretical results for the average error summarized in \cref{cor:balanced_upper_bound,thm:upper_bound_nested_tight} are upper bounds intended to observe the qualitative behavior of the constructions and are not well-suited for a quantitative comparison. 
To observe the actual behavior, we perform a numerical comparison of both code constructions, for which we set the shaping lattice $\Lambda$ as the well-known 2-dimensional hexagonal lattice scaled to meet the power constraint $P_X$. We set $q=25$ as the quantization parameter. We set $k=1$ to satisfy $n>k$, i.e. we transmit one quantized real number per device, with two uses of the AWGN channel. %
For the balanced number construction from \cref{sec:balanced_numbers}, we set $\beta=5$ and $\alpha=2$. For the coding lattice $\Lambda_1$, we generate $100$ random generator matrices and construct the corresponding codes from Section~\ref{sec:nested_lattices}. For each code construction, we generate $2\times10^6$ samples of the i.i.d. Gaussian vector $\mathbf{N}$ to experimentally evaluate the error probability.

The resulting experimental error for both schemes for per-device SNR values of $2$dB and $2.35$dB is seen in Figure~\ref*{fig:result_experimental}. We show the code constructions' minimum, maximum, and median experimental error values. As expected, the nested lattice code has a lower probability of error when there is only a single device. However, $P_e$ increases rapidly with increasing $K$. On the other hand, for the balanced numeral codes, while $P_e$ is initially higher with $K=1$, it is constant as the number of devices increases. The smaller the per-device SNR, the smaller the threshold value of $K$ above which the balanced number construction provides a lower error probability.

\section{Conclusion}

We considered the channel coding problem for AirComp for FEEL with a focus on settings with a large number of clients. We introduced a lattice-based code construction based on a balanced number system that provides competitive error rates that are stable in the number of clients. We showed superior performance compared to the well-studied nested lattice codes when the number of devices is large and gave analytical upper bounds on the average block error probabilities. Tightening the theoretical bounds for both constructions and finding a code construction with equally good properties and a provable asymptotic rate is left for future research.

\bibliographystyle{IEEEtran}
\bibliography{refs}

\begin{thebibliography}{10}
\providecommand{\url}[1]{#1}
\csname url@samestyle\endcsname
\providecommand{\newblock}{\relax}
\providecommand{\bibinfo}[2]{#2}
\providecommand{\BIBentrySTDinterwordspacing}{\spaceskip=0pt\relax}
\providecommand{\BIBentryALTinterwordstretchfactor}{4}
\providecommand{\BIBentryALTinterwordspacing}{\spaceskip=\fontdimen2\font plus
\BIBentryALTinterwordstretchfactor\fontdimen3\font minus \fontdimen4\font\relax}
\providecommand{\BIBforeignlanguage}[2]{{%
\expandafter\ifx\csname l@#1\endcsname\relax
\typeout{** WARNING: IEEEtran.bst: No hyphenation pattern has been}%
\typeout{** loaded for the language `#1'. Using the pattern for}%
\typeout{** the default language instead.}%
\else
\language=\csname l@#1\endcsname
\fi
#2}}
\providecommand{\BIBdecl}{\relax}
\BIBdecl

\bibitem{mcmahan2017communication}
B.~McMahan, E.~Moore, D.~Ramage, S.~Hampson, and B.~A. y~Arcas, ``Communication-efficient learning of deep networks from decentralized data,'' in \emph{Artificial intelligence and statistics}, 2017, pp. 1273--1282.

\bibitem{deepGC}
Y.~Lin, S.~Han, H.~Mao, Y.~Wang, and B.~Dally, ``Deep gradient compression: Reducing the communication bandwidth for distributed training,'' in \emph{International Conference on Learning Representations}, 2018.

\bibitem{QSGD}
D.~Alistarh, D.~Grubic, J.~Li, R.~Tomioka, and M.~Vojnovic, ``Qsgd: Communication-efficient sgd via gradient quantization and encoding,'' in \emph{Advances in Neural Information Processing Systems}, vol.~30, 2017.

\bibitem{signSGD}
J.~Bernstein, Y.-X. Wang, K.~Azizzadenesheli, and A.~Anandkumar, ``sign{SGD}: Compressed optimisation for non-convex problems,'' in \emph{International Conference on Machine Learning}, vol.~80, 2018, pp. 560--569.

\bibitem{ml_wireless_edge}
M.~M. Amiri and D.~Gündüz, ``Machine learning at the wireless edge: Distributed stochastic gradient descent over-the-air,'' in \emph{IEEE International Symposium on Information Theory (ISIT)}, 2019, pp. 1432--1436.

\bibitem{yang2020federated}
K.~Yang, T.~Jiang, Y.~Shi, and Z.~Ding, ``Federated learning via over-the-air computation,'' \emph{IEEE Transactions on Wireless Communications}, vol.~19, no.~3, pp. 2022--2035, 2020.

\bibitem{zhu2021over}
G.~Zhu, J.~Xu, K.~Huang, and S.~Cui, ``Over-the-air computing for wireless data aggregation in massive iot,'' \emph{IEEE Wireless Communications}, vol.~28, no.~4, pp. 57--65, 2021.

\bibitem{harnessing_interference}
M.~Goldenbaum, H.~Boche, and S.~Stańczak, ``Harnessing interference for analog function computation in wireless sensor networks,'' \emph{IEEE Transactions on Signal Processing}, vol.~61, no.~20, pp. 4893--4906, 2013.

\bibitem{gastpar_uncoded_optimal}
M.~Gastpar, ``Uncoded transmission is exactly optimal for a simple {Gaussian} “sensor” network,'' \emph{IEEE Transactions on Information Theory}, vol.~54, no.~11, pp. 5247--5251, 2008.

\bibitem{yang2022over}
H.~Yang, P.~Qiu, J.~Liu, and A.~Yener, ``Over-the-air federated learning with joint adaptive computation and power control,'' in \emph{IEEE International Symposium on Information Theory (ISIT)}, 2022, pp. 1259--1264.

\bibitem{obda}
G.~Zhu, Y.~Du, D.~Gündüz, and K.~Huang, ``One-bit over-the-air aggregation for communication-efficient federated edge learning: Design and convergence analysis,'' \emph{IEEE Transactions on Wireless Communications}, vol.~20, no.~3, pp. 2120--2135, 2021.

\bibitem{qiao2023unsourced}
L.~Qiao, Z.~Gao, Z.~Li, and D.~G{\"u}nd{\"u}z, ``Unsourced massive access-based digital over-the-air computation for efficient federated edge learning,'' in \emph{IEEE International Symposium on Information Theory (ISIT)}, 2023, pp. 2003--2008.

\bibitem{comp_over_mac}
B.~Nazer and M.~Gastpar, ``Computation over multiple-access channels,'' \emph{IEEE Transactions on Information Theory}, vol.~53, no.~10, pp. 3498--3516, 2007.

\bibitem{analog_match}
Y.~Kochman and R.~Zamir, ``Analog matching of colored sources to colored channels,'' \emph{IEEE Transactions on Information Theory}, vol.~57, no.~6, pp. 3180--3195, 2011.

\bibitem{lattices_good}
U.~Erez, S.~Litsyn, and R.~Zamir, ``Lattices which are good for (almost) everything,'' \emph{IEEE Transactions on Information Theory}, vol.~51, no.~10, pp. 3401--3416, 2005.

\bibitem{nomographic_efficient}
M.~Goldenbaum, H.~Boche, and S.~Stańczak, ``Nomographic functions: Efficient computation in clustered {Gaussian} sensor networks,'' \emph{IEEE Transactions on Wireless Communications}, vol.~14, no.~4, pp. 2093--2105, 2015.

\bibitem{compute_and_forward}
B.~Nazer and M.~Gastpar, ``Compute-and-forward: Harnessing interference through structured codes,'' \emph{IEEE Transactions on Information Theory}, vol.~57, no.~10, pp. 6463--6486, 2011.

\bibitem{local_gossip}
B.~Nazer, A.~G. Dimakis, and M.~Gastpar, ``Local interference can accelerate gossip algorithms,'' \emph{IEEE Journal of Selected Topics in Signal Processing}, vol.~5, no.~4, pp. 876--887, 2011.

\bibitem{achieving_capacity}
U.~Erez and R.~Zamir, ``Achieving 1/2 log (1+snr) on the awgn channel with lattice encoding and decoding,'' \emph{IEEE Transactions on Information Theory}, vol.~50, no.~10, pp. 2293--2314, 2004.

\bibitem{channelcomp}
S.~Razavikia, J.~M. Barros~da Silva, and C.~Fischione, ``Channelcomp: A general method for computation by communications,'' \emph{IEEE Transactions on Communications}, vol.~72, no.~2, pp. 692--706, 2024.

\bibitem{balanced_ofdm}
A.~Şahin, ``Over-the-air computation based on balanced number systems for federated edge learning,'' \emph{IEEE Transactions on Wireless Communications}, pp. 1--1, 2023.

\bibitem{you2023broadband}
L.~You, X.~Zhao, R.~Cao, Y.~Shao, and L.~Fu, ``Broadband digital over-the-air computation for wireless federated edge learning,'' \emph{IEEE Transactions on Mobile Computing}, 2023.

\bibitem{xie2023joint}
X.~Xie, C.~Hua, J.~Hong, and Y.~Wei, ``Joint design of coding and modulation for digital over-the-air computation,'' \emph{arXiv preprint arXiv:2311.06829}, 2023.

\bibitem{gupta2015deep}
S.~Gupta, A.~Agrawal, K.~Gopalakrishnan, and P.~Narayanan, ``Deep learning with limited numerical precision,'' in \emph{International conference on machine learning}, 2015, pp. 1737--1746.

\bibitem{croci2022stochastic}
M.~Croci, M.~Fasi, N.~J. Higham, T.~Mary, and M.~Mikaitis, ``Stochastic rounding: implementation, error analysis and applications,'' \emph{Royal Society Open Science}, vol.~9, no.~3, p. 211631, 2022.

\bibitem{erez2005capacity}
U.~Erez, S.~Shamai, and R.~Zamir, ``Capacity and lattice strategies for canceling known interference,'' \emph{IEEE Transactions on Information Theory}, vol.~51, no.~11, pp. 3820--3833, 2005.

\bibitem{forney2000sphere}
G.~Forney, M.~Trott, and S.-Y. Chung, ``Sphere-bound-achieving coset codes and multilevel coset codes,'' \emph{IEEE Transactions on Information Theory}, vol.~46, no.~3, pp. 820--850, 2000.

\bibitem{loeliger}
H.-A. Loeliger, ``Averaging bounds for lattices and linear codes,'' \emph{IEEE Transactions on Information Theory}, vol.~43, no.~6, pp. 1767--1773, 1997.

\end{thebibliography}

\end{document}